\documentclass[]{article}
\usepackage[left=2cm,right=2cm, top=2.4cm,bottom=2.4cm,bindingoffset=0cm]{geometry}
\usepackage{tikz, titlesec, tikz-cd}
\usetikzlibrary{calc,intersections,through,backgrounds}
\usepackage{amsmath, amssymb}
\usepackage{amsthm, mathtools}

\usepackage{array}
\newcolumntype{?}{!{\vrule width 1pt}}

\usepackage{ifthen}

\usepackage{faktor}

\mathtoolsset{showonlyrefs}

\usepackage{makecell}

\usepackage[numbers, sort&compress]{natbib}

\usepackage[bottom]{footmisc}

\usepackage{hyperref}

\definecolor{color1}{RGB}{127,201,127}
\definecolor{color2}{RGB}{190,174,212}
\definecolor{color3}{RGB}{253,192,134}
\definecolor{color4}{RGB}{255,255,153}

\usetikzlibrary{positioning}
\usetikzlibrary{decorations.text}
\usetikzlibrary{decorations.pathmorphing}
\usetikzlibrary{decorations.markings}

\usepackage{tkz-euclide}

\newtheorem{lemma}{Lemma}[section]

\newtheorem{theorem}[lemma]{Theorem}

\theoremstyle{definition}
\newtheorem{remark}[lemma]{Remark}

\newcommand{\Z}{\mathbb{Z}}

\newcommand{\R}{\mathbb{R}}

\usepackage{bbm}

\newcommand{\Ca}{C_a}
\newcommand{\Cb}{C_b}
\usepackage{titlesec}

\titleformat*{\section}{\large\bfseries}
\titleformat*{\subsection}{\Large\bfseries}
\titleformat*{\subsubsection}{\large\bfseries}
\titleformat*{\paragraph}{\large\bfseries}
\titleformat*{\subparagraph}{\large\bfseries}

\title{Folding of quadrilaterals, zigzags, and Arnold-Liouville integrability}
\author{Anton Izosimov\thanks{
Department of Mathematics,
University of Arizona and School of Mathematics \& Statistics, University of Glasgow;
e-mail: {\tt Anton.Izosimov@glasgow.ac.uk}
} }
\date{}
\begin{document}

\tikzset{->-/.style={decoration={
  markings,
  mark=at position .7 with {\arrow{>}}},postaction={decorate}}}
  
\usetikzlibrary{angles, quotes}

\maketitle

\abstract{
We put Darboux's porism on folding of quadrilaterals, as well as closely related Bottema's zigzag porism, in the context of Arnold-Liouville integrability. 
}
\section{Introduction}

One of the historically first manifestations of integrability is \emph{Poncelet's  porism}, also known as \emph{Poncelet's  closure theorem}. Poncelet's theorem says that if a planar $n$-gon is inscribed in a conic $C_1$ and circumscribed about another conic $C_2$, then any point of $C_1$ is a vertex of such an $n$-gon, see Figure \ref{Fig:poncelet}. 
The two arguably most standard proofs of this theorem are based, respectively, on complex and symplectic geometry. The complex proof goes roughly as follows. One can identify the space of tangents dropped from a point on $C_1$ to $C_2$ with an elliptic curve. The successive sides of a polygon  inscribed in $C_1$ and circumscribed about $C_2$ are points on that curve related to each by a fixed translation. This polygon closes up if and only if the translation vector is a torsion point on the elliptic curve. Whether or not that is the case depends only on $C_1$ and $C_2$, but not on the initial point, so all polygons  inscribed in $C_1$ and circumscribed about $C_2$ will close up after the same number of steps \cite{griffiths1977poncelet}. 

\begin{figure}[b]
\centering
\includegraphics[width = 4 cm]{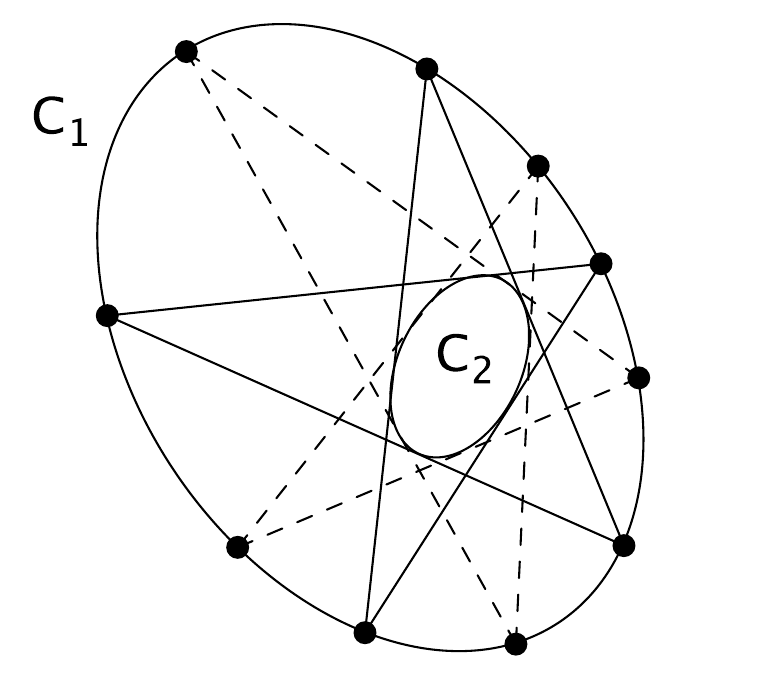}
\caption{Every point of $C_1$ is a vertex of a pentagon inscribed in $C_1$ and circumscribed about $C_2$.}\label{Fig:poncelet}
\end{figure}

The second, symplectic, proof is based on the fact that any two generic conics can be mapped, by a projective transformation, to confocal conics. In the confocal case, a polygon  inscribed in $C_1$ and circumscribed about $C_2$  can be identified with a billiard trajectory in $C_1$. The billiard in a conic is an integrable system, and any two polygons  inscribed in $C_1$ and circumscribed about $C_2$ correspond to trajectories belonging to the same level set of the first integral. Hence, by Arnold-Liouville theorem, if one of the trajectories is periodic with period $n$, then so is the other one, cf. \cite{levi2007poncelet}.

A lesser known relative of Poncelet's porism is \emph{Darboux's porism on folding of quadrilaterals}. \emph{Folding} of a vertex of a planar polygon is the reflection of that vertex is the diagonal joining its neighbors, see Figure \ref{Fig:darboux}. Darboux's porism says that if a sequence of alternating foldings of adjacent vertices restores, after $2n$ steps, the initial polygon, then this is the case for any polygon with the same side lengths. For example, folding any polygon with side lengths $1, 3, 3\sqrt{5}, 5$ six times, we come back to the initial polygon, see \cite[Figure 2]{Izm}.

\begin{figure}[t]
\centering
\includegraphics[width = 6.5 cm]{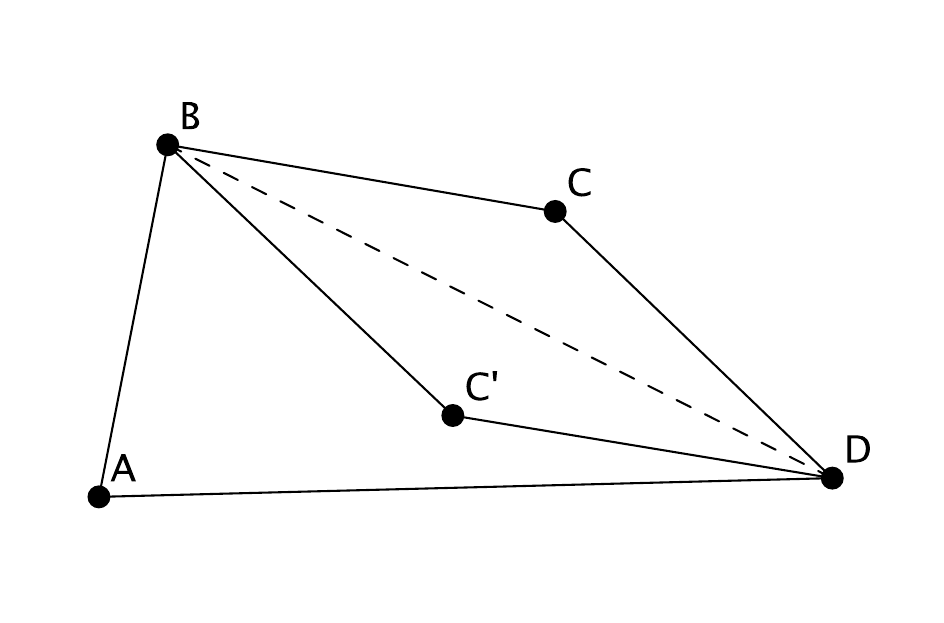}
\caption{Folding of the vertex $C$ of a quadrilateral $ABCD$. Its new position is $C'$.}\label{Fig:darboux}
\end{figure}

Just like Poncelet's porism, Darboux's theorem can be proved using elliptic curves. Specifically, one shows that the complexified moduli space of quadrilaterals with fixed side length is an elliptic curve. Composition of foldings at adjacent vertices amounts to a translation on that curve. Whether or not a sequence of foldings restores the initial polygon depends on whether the translation vector is a torsion point and is independent on the particular choice of a quadrilateral \cite{Izm}.

What currently seems to be missing in the literature is a symplectic proof of Darboux's theorem. We provide such a proof in the present paper. Specifically, we show that, in an appropriate sense, Darboux's folding is  Arnold-Liouville integrable, and deduce Darboux's porism. 

Furthermore, we extend these results to Bottema's \emph{zigzag porism} \cite{bottema1965schliessungssatz}, which can be stated as follows. Let $\Ca$ and $\Cb$ be two circles such that there exists a unit equilateral $2n$-gon
whose odd-indexed vertices lie on $\Ca$ and even-indexed vertices lie on
$\Cb$. Then there exist infinitely many such $2n$-gons. The zigzag porism is equivalent to Darboux's porism when the circles are coplanar \cite{csikos2000remarks}, but is in fact valid for any two circles in $\R^3$~\cite{black1974theorem}. We construct the underlying Arnold-Liouville integrable system in this more general setting. 
\medskip

{\bf Acknowledgments.} The author learned about the zigzag porism from the referee of the first version of this paper. The author is grateful to Max Planck Institute for Mathematics in
Bonn for its hospitality and financial support. This work was partially supported by the Simons Foundation through the Travel Support for Mathematicians program. Figures were created with help of software package Cinderella.

\section{Arnold-Liouville integrability of folding}\label{sec:alif}

Let $\mathcal P$ be the space of quadrilaterals $ABCD$ with fixed side lengths, considered up to orientation-preserving isometries. There is abundant literature on the topology of such spaces for polygons with any number of vertices, see \cite{kapovich1995moduli} and references therein. The space $\mathcal P$ is a smooth manifold assuming that there is no linear combination of side lengths with coefficients $\pm 1$ which is equal to zero \cite[Lemma 2]{kapovich1995moduli}. In the case of quadrilaterals, this manifold is diffeomorphic to a circle or disjoint union of two circles, see \cite[Theorem 1]{kapovich1995moduli}. These circles are distinguished by the sign of the oriented area and are interchanged by an orientation-reversing isometry, cf. \cite[Section 10]{kapovich1995moduli}.

Denote by $F_B \colon \mathcal P \to \mathcal P$ folding of the vertex $B$. This mapping is well-defined assuming that the vertices $A$ and $C$ cannot come together. This holds provided that the side lengths satisfy at least one of the following conditions: $|AB| \neq |BC|$ or $|AD| \neq |CD|$. 
Likewise, let $F_C \colon \mathcal P \to \mathcal P$ be folding of $C$, and let $F := F_C \circ F_B$ be the composition of the two foldings. Darboux's porism says that if $F^n(P) = P$ for some quadrilateral $P \in \mathcal P$, then $F^n$ is the identity mapping on $\mathcal P$.  We shall prove this by establishing Arnold-Liouville integrability of the mapping $F$.

Clearly, $F$ cannot be Arnold-Liouville integrable on the space $\mathcal P$ of quadrilaterals with fixed side lengths, as the latter space is one-dimensional. So, we consider a bigger space ${\mathcal{P'}}$ of quadrilaterals with fixed lengths of the sides $AB, BC, CD$, again considered up to orientation-preserving isometries. This space is diffeomorphic to a two-dimensional torus and is parametrized by the oriented angles $\angle ABC$ and $\angle BCD$. The squared length of the side $AD$ is a smooth function of the torus $\mathcal P'$. The space $\mathcal P$ of quadrilaterals with fixed lengths of all sides is a level set of that function.

\begin{theorem}\label{dap}
The folding mapping  $F = F_C \circ F_B$ is Arnold-Liouville integrable on the moduli space $\mathcal P'$ of quadrilaterals $ABCD$ with fixed lengths of the sides $AB, BC, CD$.
\end{theorem}

\begin{proof}
Folding does not affect side lengths. In particular, $|AD|^2$ is a first integral of $F$. Furthermore, the map $F \colon \mathcal P' \to \mathcal P'$ has an invariant symplectic structure given by
$$
\Omega := d\angle ABC \wedge d\angle BCD.
$$
To show invariance, consider, for instance, folding of the vertex $C$ depicted in Figure \ref{Fig:darboux}. The pullback of the symplectic form $\Omega$ by this map is
$$
F_C^*\Omega =d \angle ABC'  \wedge d  \angle BC'D = d(\angle ABC - 2\angle CBD) \wedge d(2\pi -  \angle BCD) = -\Omega -2d\angle CBD \wedge d \angle BCD.
$$
Furthermore, since the side lengths $|BC|$ and $|CD|$ are fixed, the angle $\angle CBD$ is a function of the angle $\angle BCD$ and is independent of the angle $\angle ABC$. So,
$
d\angle CBD \wedge d \angle BCD = 0,
$
implying
$$
F_C^*\Omega = - \Omega,
$$
i.e., the form $\Omega$ is \emph{anti-invariant} under a single folding, and hence invariant under $F$.
\end{proof}


\section{Darboux's porism}

\begin{theorem}[{Darboux's porism}] \label{dpt}
Assume we are given a quadrilateral which restores its initial shape after $2n$ alternating foldings at adjacent vertices. Suppose its side lengths are such that no linear combination of them with coefficients $\pm 1$ is equal to zero. Then any quadrilateral with  the same side lengths  restores its initial shape after $2n$ alternating foldings at adjacent vertices. 
\end{theorem}
\begin{remark}
The condition on linear combinations of side length cannot be avoided. Consider, for instance a quadrilateral with all four vertices along a line, shown in Figure \ref{Fig:darboux2}. Here we have $|AB| =2$, $|BC| = 1 $, $|CD| = 2$, $|AD| = 3$. Clearly, this quadrilateral is invariant under any folding. However, that is not so for a generic quadrilateral with side lengths $2,1,2,3$.
\end{remark}

\begin{proof}[Proof of Theorem \ref{dpt}]
The assumption on linear combinations of  side lengths ensures that the moduli space $\mathcal P$ of quadrilaterals with such side lengths is a regular level set of the function $|AD|^2$ on the moduli space $\mathcal P'$ of polygons with fixed lengths of  $AB, BC, CD$. We are given that there is a quadrilateral $P \in \mathcal P$ on that level set such that $F^n(P) = P$. So, by Arnold-Liouville integrability of $F$, we have that $F^n$ is the identity on the connected component of $\mathcal P$ containing $P$. Moreover, since there are at most two components, and they are interchanged by an orientation-reversing isometry which commutes with foldings, we must have that $F^n$ is the identity of the whole of $\mathcal P$, as desired.
\end{proof}

\section{A remark on polygons with more vertices}
The $F$-invariant  symplectic form on the moduli space $\mathcal P'$ of quadrilaterals with  fixed lengths of the sides $AB, BC, CD$ induces an $F$-invariant non-vanishing $1$-form on any non-singular level set of the first integral $|AD|^2$, i.e., on the moduli space $\mathcal P$ of quadrilaterals  with fixed side lengths.  The existence of this $1$-form is at heart of Arnold-Liouville theorem. It can be shown that, up to a constant factor, this form is given by
$$
\frac{d\angle ABC}{\text{area of}\, \triangle ACD}
$$
This expression is invariant under cyclic permutation of vertices, up to sign. Likewise, the expression 
$$
\frac{d\phi_{i+2} \wedge \dots \wedge d \phi_{i-2}}{\text{area of the triangle formed by vertices} \, i-1,i,i+1},
$$
 where $\phi_i$ is the angle subdued at $i$th vertex (the indices are understood cyclically, modulo $n$), gives a volume form on the moduli space of $n$-gons with fixed side lengths which is anti-invariant under each folding and hence invariant under an even number of foldings. However, for $n > 4$, this does not imply any kind of integrable behavior.   Moreover, already for pentagons a random sequence of foldings has dense orbits on the moduli space~$\mathcal P$~\cite{cantat2023random}.

\begin{figure}[t]
\centering
\includegraphics[width = 8 cm]{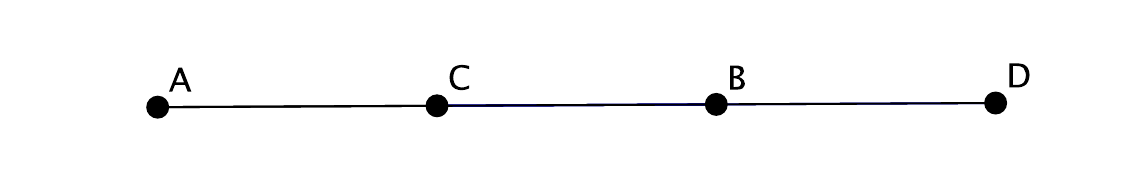}
\caption{A degenerate polygon.}\label{Fig:darboux2}
\end{figure}

\section{The zigzag porism}
Let $\Ca$ and $\Cb$ be two circles in $\R^3$. A \emph{zigzag} between $\Ca$ and $\Cb$ is an equilateral polygon
whose odd-indexed vertices lie on $\Ca$ and even-indexed vertices lie on
$\Cb$. The zigzag porism says that if there exists a closed $2n$-gonal zigzag between $\Ca$ and $\Cb$, then any zigzag  between $\Ca$ and $\Cb$ with the same edge length is also a closed $2n$-gon \cite{bottema1965schliessungssatz, black1974theorem}, see Figure \ref{Fig:zigzag}.

\begin{figure}[t]
\centering
\includegraphics[width = 6.5 cm]{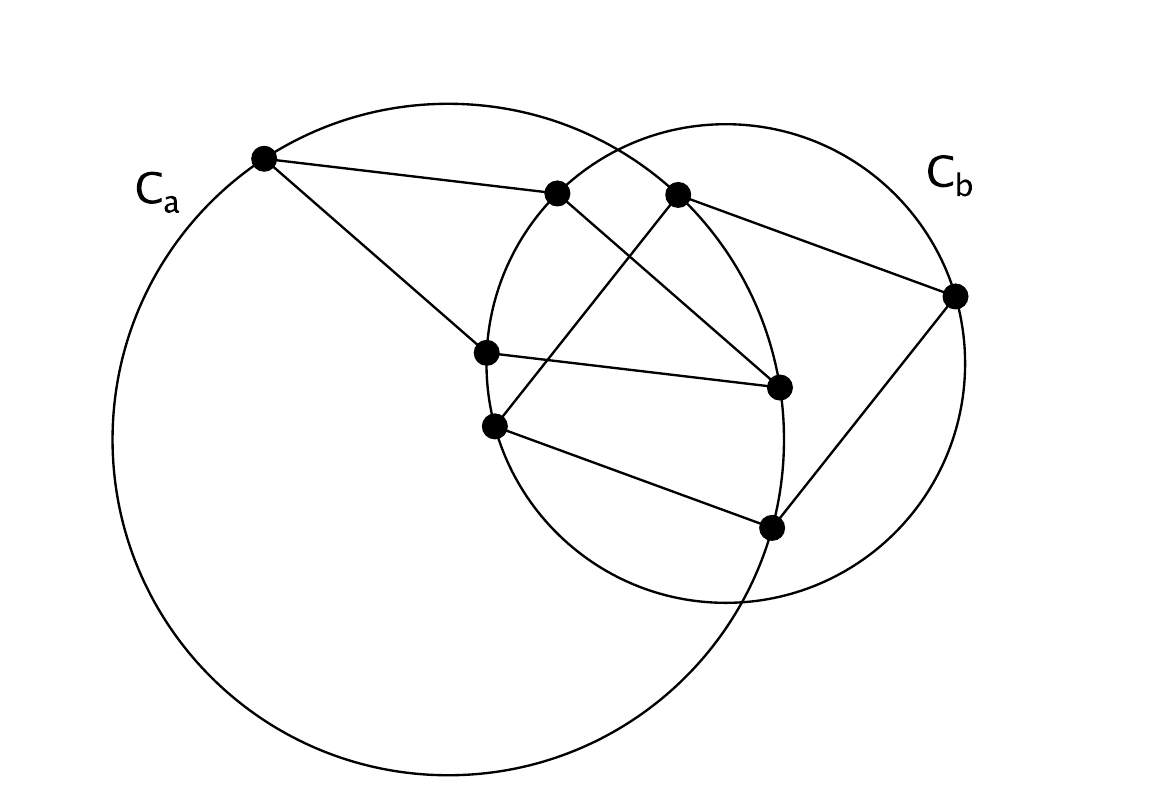}
\caption{The zigzag porism: all zigzags with the same edge length close after the same number of steps.}\label{Fig:zigzag}
\end{figure}

\begin{figure}[t]
\centering
\includegraphics[width = 7 cm]{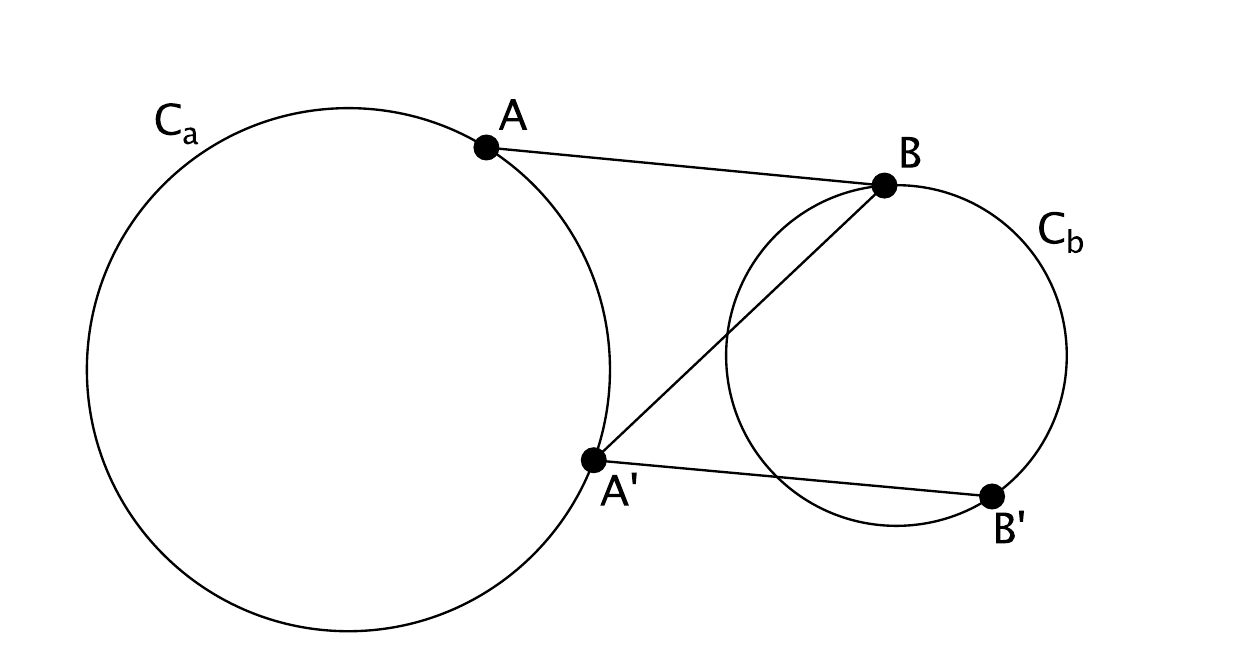}
\caption{The zigzag map $Z \colon (A,B) \mapsto (A', B')$.}\label{Fig:zigzagb}
\end{figure}

\begin{figure}[b]
\centering
\includegraphics[width = 7 cm]{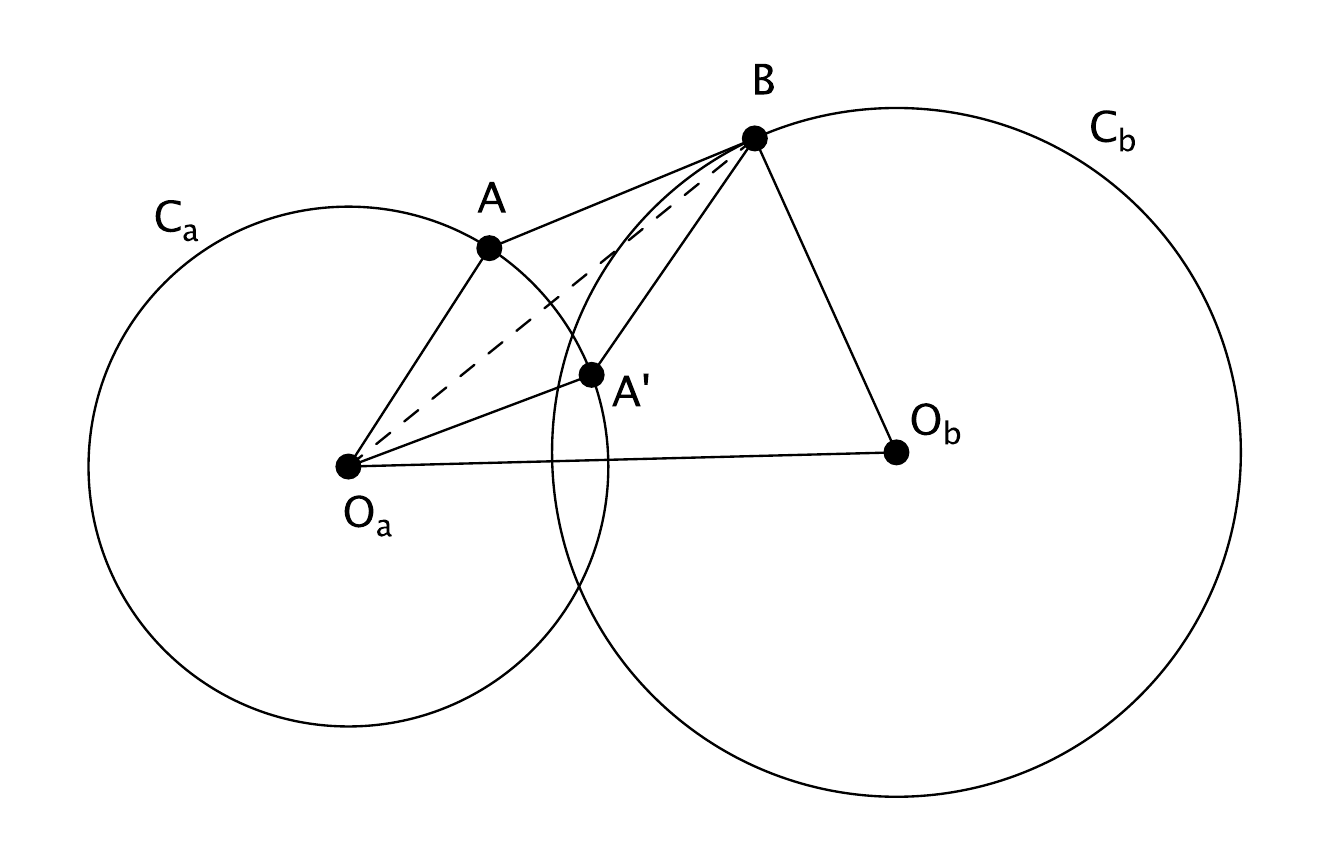}
\caption{Two successive legs $AB$, $BA'$ of a zigzag are related by folding.}\label{Fig:zigzagf}
\end{figure}

A zigzag between two circles $\Ca$, $\Cb$ may be built by iterating the \emph{zigzag map} $Z \colon \Ca \times \Cb \to  \Ca \times \Cb$ which sends a pair $A \in \Ca$, $B \in \Cb$ to a pair $A' \in \Ca$, $B' \in \Cb$ such that $|A'B'| = |A'B| = |AB|$, see Figure \ref{Fig:zigzagb}. This map is a composition of two involutions, namely $(A,B) \mapsto (A',B)$, where $|A'B| = |AB|$, and $(A',B) \mapsto (A',B')$, where $|A'B'| = |A'B|$. Observe that, in the case when the circles $\Ca, \Cb$ are coplanar, these involutions are just foldings of the quadrilateral $O_aA B O_b$, where $O_a$, $O_b$ are centers of $\Ca, \Cb$, at $A$ and $B$ respectively, see Figure \ref{Fig:zigzagf}. So, the planar case of the zigzag porism is equivalent to Darboux's porism \cite{csikos2000remarks}. Here we show that the integrability result carries over to the spatial situation:

\begin{figure}[t]
\centering
\includegraphics[width = 6 cm]{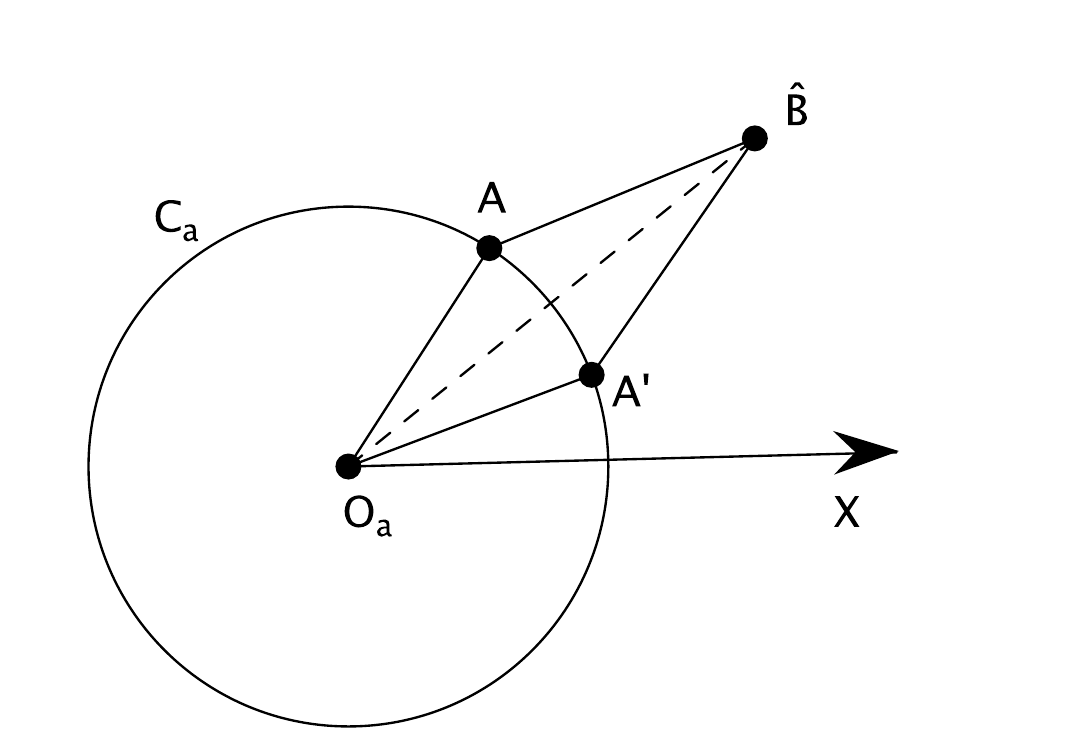}
\caption{The involution $(A,B) \mapsto (A', B)$ takes the form  $d \phi_a \wedge d\phi_b$ to  $-d \phi_a \wedge d\phi_b$.}\label{Fig:zigzagpf}
\end{figure}
\begin{theorem}\label{zzp}
The zigzag map $Z$ is Arnold-Liouville integrable for any circles $\Ca$, $\Cb$ in $\R^3$.
\end{theorem}
\begin{proof}
By definition, the map $Z \colon (A,B) \mapsto (A', B')$ preserves the squared distance between $A$ and $B$. So, it suffices to find an area form on  $\Ca \times \Cb$ invariant under $Z$. Let $\phi_a, \phi_b \in \R/2\pi \Z$ be standard angular parameters on $\Ca$, $\Cb$. We will prove that the form $d \phi_a \wedge d\phi_b$ on $\Ca \times \Cb$ is preserved by $Z$. To that end, it suffices to establish anti-invariance of that form with respect to the involutions whose composition gives $Z$. Furthermore, since those involutions are related to each other by interchanging the roles of the circles $C_a, C_b$, it is sufficient to consider the involution $(A,B) \mapsto (A', B)$ defined by the condition $|A'B| = |AB|$, where $A, A' \in C_a$. Let $\hat B$ be the orthogonal projection of $B$ onto the plane containing $C_a$. Then $|A\hat B| = |A'\hat B|$, see Figure \ref{Fig:zigzagpf}. Here $O_aX$ is the reference direction used to define the angular coordinated $\phi_a$ on $\Ca$. We have
$$
\angle XO_a A + \angle XO_a A' = 2 \angle XO_a \hat B.
$$
So, the sum on the left only depends on the position of the point $B$ but not $A$. Therefore, in coordinates $\phi_a, \phi_b$, the involution $(A,B) \mapsto (A', B)$ has the form
$$
(\phi_a, \phi_b) \mapsto (f(\phi_B) - \phi_a, \phi_b)
$$
for a certain smooth function $f$. So, the form $ d \phi_a \wedge d\phi_b $ is indeed anti-invariant under this involution.
\end{proof}
In terms of the map $Z$, the zigzag porism says that if an orbit of $(A,B) \in \Ca \times \Cb$ under $Z$ is $n$-periodic, then the same holds for any $(A',B') \in \Ca \times \Cb$ with $|A'B'| = |AB|$. This is derived from Theorem \ref{zzp} in the same way as Darboux's porism is obtained from Theorem \ref{dap}.




\bibliographystyle{plain}
\bibliography{fold.bib}

\begin{thebibliography}{1}

\bibitem{black1974theorem}
W.L. Black, H.C. Howland, and B.~Howland.
\newblock A theorem about zig-zags between two circles.
\newblock {\em The American Mathematical Monthly}, 81(7):754--757, 1974.

\bibitem{bottema1965schliessungssatz}
O.~Bottema.
\newblock Ein schliessungssatz f{\"u}r zwei kreise.
\newblock {\em Elemente der Mathematik}, 20:1--7, 1965.

\bibitem{cantat2023random}
S.~Cantat and R.~Dujardin.
\newblock Random dynamics on real and complex projective surfaces.
\newblock {\em Journal f{\"u}r die reine und angewandte Mathematik (Crelles
  Journal)}, 2023(802):1--76, 2023.

\bibitem{csikos2000remarks}
B.~Csik{\'o}s and A.~Hrask{\'o}.
\newblock Remarks on the zig-zag theorem.
\newblock {\em Periodica Mathematica Hungarica}, 39(1):201--211, 2000.

\bibitem{griffiths1977poncelet}
P.~Griffiths and J.~Harris.
\newblock A {P}oncelet theorem in space.
\newblock {\em Commentarii Mathematici Helvetici}, 52(2):145--160, 1977.

\bibitem{Izm}
I.~Izmestiev.
\newblock Deformation of quadrilaterals and addition on elliptic curves.
\newblock {\em Moscow Mathematical Journal}, 23:205--242, 2023.

\bibitem{kapovich1995moduli}
M.~Kapovich and J.~Millson.
\newblock On the moduli space of polygons in the {E}uclidean plane.
\newblock {\em Journal of Differential Geometry}, 42(2):430--464, 1995.

\bibitem{levi2007poncelet}
M.~Levi and S.~Tabachnikov.
\newblock The {P}oncelet grid and billiards in ellipses.
\newblock {\em The American Mathematical Monthly}, 114(10):895--908, 2007.

\end{thebibliography}

\end{document}